\documentclass{cccg24}
\usepackage{comment}
\usepackage{array,amsmath}
\usepackage{amssymb}
\usepackage{graphicx}
\usepackage{caption}
\usepackage{url}
\usepackage{subcaption}
\makeatletter
\usepackage{authblk}
\usepackage{algorithm}
\usepackage{algorithmic}

\newtheorem{corollary}[theorem]{Corollary}
\makeatother

\title{Multirobot Watchman Routes in a Simple Polygon}
        \author {Joseph S. B. Mitchell\thanks{Department of Applied Mathematics and Statistics, Stony Brook University, \texttt{joseph.mitchell@stonybrook.edu}} \hspace{0.1\textwidth} Linh Nguyen\thanks{Department of Applied Mathematics and Statistics, Stony Brook University, \texttt{linh.nguyen.1@stonybrook.edu}}}
\begin{document}
	\maketitle
\begin{abstract}
    The well-known \textsc{Watchman Route} problem seeks a shortest route in a polygonal domain from which every point of the domain can be seen. In this paper, we study the cooperative variant of the problem, namely the \textsc{$k$-Watchmen Routes} problem, in a simple polygon $P$. We look at both the version in which the $k$ watchmen must collectively see all of $P$, and the quota version in which they must see a predetermined fraction of $P$'s area.
	
    We give an exact pseudopolynomial time algorithm for the \textsc{$k$-Watchmen Routes} problem in a simple orthogonal polygon $P$ with the constraint that watchmen must move on axis-parallel segments, and there is a given common starting point on the boundary. Further, we give a fully polynomial-time approximation scheme and a constant-factor approximation for unconstrained movement. For the quota version, we give a constant-factor approximation in a simple polygon, utilizing the solution to the (single) \textsc{Quota Watchman Route} problem.
\end{abstract}
\section{Introduction}
In 1973, Victor Klee introduced the \textsc{Art Gallery} problem: given an art gallery with $n$ walls (a polygon $P$), determine the minimum number of stationary guards at points within $P$ such that every point of $P$ can be seen by at least one guard point. The \textsc{Art Gallery} problem and its many variants have since been the subject of a large body of research in computational geometry and algorithms. 

When guards are mobile, a single guard suffices to see a connected domain; thus, we are interested in finding routes for one or more guards that optimize some aspects of the guard(s)' movement (e.g., path lengths, the number of turns, etc). The problem of minimizing the distance that one guard must travel to see the entire polygon is the \textsc{Watchman Route} problem (WRP). Chin and Ntafos~\cite{CN} introduced the WRP, proved NP-hardness in polygons with holes (see \cite{dumitrescu2012watchman}) and gave an $O(n)$ algorithm for simple orthogonal polygons. In (general) simple polygons, there are exact polynomial-time algorithms; the current best running times are $O(n^3\log n)$ for the \textit{anchored} version (a starting point $s$ which the route must pass through is given) and $O(n^4\log n)$ for the \textit{floating} version (no starting point is given)~\cite{dror2003touring}.

In some settings, complete coverage might not be feasible or necessary, thus we are also interested in computing a shortest route that sees at least an area of $A \ge 0$ within $P$. This is known as the \textsc{Quota Watchman Route} problem (QWRP), introduced in~\cite{huynh_et_al:LIPIcs.SWAT.2024.27}. In contrast to the tractable WRP, the QWRP is (weakly) NP-hard, but a fully polynomial-time approximation scheme (FPTAS) is known. Any results about the QWRP can be adapted to the WRP by simply letting $A$ be equal to the area of $P$.

We consider the generalization to multiple agents of both the WRP and the QWRP, namely the \textsc{$k$-Watchmen Routes} problem ($k$-WRP) and the \textsc{Quota $k$-Watchmen Routes} problem (Q$k$-WRP), with the objective of minimizing the length of the longest path traveled by any one watchman. Even in a simple polygon, when no starting points are specified (so, we are to determine the best starting locations), both problems are NP-hard to approximate within any multiplicative factor~\cite{packer2008computing}.

We thus focus on the (boundary) anchored version, in which a team of robots or searchers enter a domain $P$ through a door on its boundary to search for a stationary target, which may be randomly distributed within the domain; the objective is to plan for an optimal collective effort to guarantee at least a certain probability of detection (1 in the $k$-WRP and some $p\in[0,1]$ in the Q$k$-WRP). We consider the number, $k$, of robots to be fixed and relatively small, as in most practical situations it is infeasible to employ arbitrarily many watchmen/robots/agents. We present, for any fixed $k$, a pseudopolynomial-time (polynomial in the number $n$ of vertices of $P$ and the length of the longest edge of $P$) exact algorithm to solve the anchored $k$-WRP in a simple orthogonal (integral coordinate) polygon $P$ under L1 distance. The pseudopolynomial-time exact algorithm is the basis for the FPTAS for L1 distance and the $(\sqrt{2} + \varepsilon)$-approximation for L2 distance. For the Q$k$-WRP, we give polynomial-time constant-factor approximations in a simple polygon. While we restrict ourselves to the anchored version, we achieve better approximation factors for any (fixed) $k$ than the ones Nilsson and Packer proposed for the case $k=2$ in~\cite{nilsson2023approximation}. 

\section{Preliminaries}
Let $P$ be a \textit{simple polygon}, i.e. a simply connected subset of $\mathbb{R}^2$. Denote by $\partial P$ the boundary of $P$, a polygonal chain that does not self-intersect consisting of $n$ vertices $v_1, v_2, \ldots, v_n$, which we assume to have integer coordinates. A simple polygon is \textit{orthogonal} if the internal angle at every vertex is either 90 (convex vertex) or 270 degrees (reflex vertex).

For a point $x\in P$, its \textit{visibility region}, denoted by $V(x)$, is the set of all points $y$ such that the segment $xy$ does not intersect with the exterior of $P$: we say $x$ and $y$ and see each other. For an arbitrary set $X\subseteq P$, the visibility region of $X$, $V(X)$, is the set of all points that are seen by some point in $X$. When $X$ is either a point or a line segment, $V(X)$ is necessarily a simple subpolygon of $P$ with at most $n$ vertices and can be computed in $O(n)$ time~\cite{guibas1986linear, toth2017handbook}. We use $|\cdot|$ to denote Euclidean measure of geometric objects (e.g., length or area).

The first problem we investigate is the anchored $k$-WRP, where the polygon $P$ is orthogonal and movements of the watchmen are rectilinear (L1 distance). Given a simple orthogonal polygon $P$ and a starting point $s \in \partial P$, we compute $k$ tours $\{\gamma_i\}$ within $P$ consisting of horizontal and vertical segments, all starting from $s$ such that $\bigcup\limits_{i=1,\ldots,k}V(\gamma_i) = P$ and $\max\limits_{i=1,\ldots,k}|\gamma_i|$ is minimized. We also assume that the coordinates of the vertices of $P$ are integers. It is known that even for $k = 2$, the general $k$-WRP in a simple polygon is (weakly) NP-hard via a simple reduction from \textsc{Partition} \cite{mitchell1991watchman}. The reduction can be easily modified to show that our version is also NP-hard. The second problem, Q$k$-WRP, generalizes the first to $\left|\bigcup\limits_{i = 1,\ldots, k}V(\gamma_i)\right| \ge A$ for some $0 \le A \le |P|$. The fraction of area seen, $\frac{A}{|P|}$, can be interpreted as the probability that the watchmen detect a target uniformly distributed in $P$. We consider the Q$k$-WRP in a simple polygon, where the watchmen have unrestricted movement (not limited to horizontal and vertical).

\section{$k$-Watchmen in a Simple Orthogonal Polygon}
\label{sect:kwatchman}

\paragraph{Dynamic programming exact algorithm} A \textit{visibility cut} $c_i$ with respect to the starting point $s$ is a chord obtained from extending the edge $e$ incident on a reflex vertex, $v_i$, where $e$ is the edge whose extension creates a convex vertex at $v_i$ in the subpolygon containing $s$. The other subpolygon (not containing $s$) is the \textit{pocket} induced by $c_i$. Not all reflex vertices induce a visibility cut. An \textit{essential cut} is a visibility cut whose pocket does not fully contain any other pocket (Figure~\ref{fig:essential_cuts}). In general, essential cuts may intersect with each other.

    \begin{figure}[!h]
		\centering
		\includegraphics[scale=1.3]{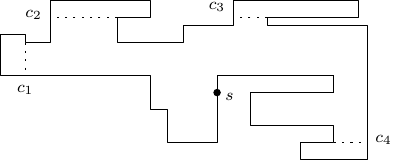}
		\caption{The essential cuts (dashed).}
		\label{fig:essential_cuts}
    \end{figure}

\begin{lemma}
    $\bigcup\limits_{i=1,\ldots,k}V(\gamma_i) = P$ if and only if $\{\gamma_i\}$ collectively visit all essential cuts of $P$.
\end{lemma}

\begin{proof}
    The lemma is simply an extension of the well known fact: a single tour sees all of $P$ if and only if it visits all essential cuts \cite{carlsson1999finding, chin1991shortest, dror2003touring}.
\end{proof}
Denote by $C_i$ the set of essential cuts visited by $\gamma_i$.
\begin{corollary}
    \label{lem:shortest}
    There exists an optimal solution $\{\gamma_i\}$ such that for any $i$, $\gamma_i$ is the shortest route to visit all cuts in $C_i$ and $s$ in the order in which they appear around $\partial P$.
\end{corollary}

%As a corollary, we can identify each route $\gamma_i$ with a subset of essential cuts that it is responsible for. Denote by $C_i$ 

Consider the decomposition of $P$ into rectangular cells by the maximal (within $P$) extensions of all edges, as well as a horizontal and vertical line through $s$; this is known as the Hanan grid (Figure~\ref{fig:HanagridinP}).
	
	\begin{figure}[!h]
		\centering
		\includegraphics[scale=1.3]{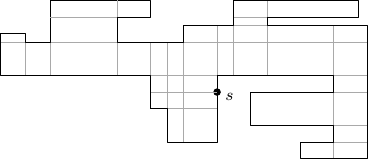}
		\caption{The Hanan grid formed by extensions of all edges in $P$.}
		\label{fig:HanagridinP}
	\end{figure}

 \begin{lemma}
 \label{lem:discrete_viewpoints}
     There exists an optimal solution $\{\gamma_i\}$ within the Hanan grid.
 \end{lemma}
 \begin{proof}
     Given an optimal solution $\{\gamma_i\}$, let $C_i = \{c_{i1}, \ldots, c_{ij}\}$ (in order around $\partial P$) and $p_{i1}, \ldots, p_{ij}$ be the point where $\gamma_i$ first makes contact with $c_{i1}, \ldots, c_{ij}$. Denote by $L1_P(x,y)$ a geodesic L1 shortest path between $x$ and $y$, a rectilinear shortest path constrained to stay within $P$. (For an overview of geodesic shortest paths in both L1 and L2 metrics, see~\cite{mitchell2000geometric}.)
     
     First, note that for every $i$, we may replace $\gamma_i$ with a concatenation of geodesic L1 shortest paths, namely $\gamma_i:=L1_p(s, p_{i1})\cup L1_P(p_{i1}, p_{i2})\cup \ldots\cup L1_P(p_{ij}, s)$ without increasing $\max\limits_{i = 1,\ldots,k}\{|\gamma_i|\}$ while maintaining visibility coverage of $P$.
     
     We argue that $L1_p(s, p_{i1})$ is a geodesic L1 shortest path from $s$ to $c_{i1}$. Suppose to the contrary, that geodesic L1 shortest paths from $s$ to $c_{i1}$ make contact with $c_{i1}$ at $p'_{i1} \ne p_{i1}$ (all geodesic L1 shortest paths from a point to a segment have the same endpoint). Due to orthogonality $|L1_P(s,p'_{i1})| + |p'_{i1}p_{i1}| = |L1_P(s,p_{i1})|$, which means $|L1_P(s,p_{i1})| + |L1_P(p_{i1}, p_{i2})| = |L1_P(s,p'_{i1})| + |p'_{i1}p_{i1}| + |L1_P(p_{i1}, p_{i2})| \ge |L1_P(s,p'_{i1})| + |L1_P(p'_{i1}, p_{i2})|$. This implies $\gamma_i$ should take a geodesic L1 shortest path from $s$ to $c_{i1}$, and it suffices to find such a path within the Hanan grid. By a straightforward inductive argument, we can show the same for any portion of $\gamma_i$ between any two essential cuts.
 \end{proof}

 Corollary~\ref{lem:shortest} and Lemma~\ref{lem:discrete_viewpoints} allow us to reduce the problem to that of finding a set of grid points on the essential cuts for which each route is responsible. Then, each route is simply the concatenation of L1 shortest paths between those points. Let $\{c_1, c_2, \ldots, c_m\}$ be the set of essential cuts in order around $\partial P$ ($s$ lies between $c_1$ and $c_m$). We define each subproblem $(c_j, p_1, l_1, \ldots, p_k, l_k)$ by an essential cut $c_j$, $k$ Hanan grid points $p_1, \ldots, p_k$ on essential cuts $c_1, \ldots, c_j$ (and $s$) and $k$ integers $l_1, \ldots, l_k$. Refer to Figure~\ref{fig:subproblem} for an illustration. Subproblem $(c_j, p_1, l_1, \ldots, p_k, l_k)=$ TRUE if and only if there exists a collection of $k$ paths $\Gamma_1, \ldots, \Gamma_k$ collectively visiting all essential cuts from $c_1$ up to $c_j$ such that
 \begin{itemize}
     \item $\Gamma_i$ starts at $s$, ends at $p_i$,
     \item $|\Gamma_i| = l_i$.
 \end{itemize}
 \begin{figure}[h]
    \centering
    \includegraphics[width=0.5\textwidth]{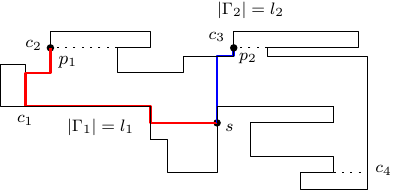}
    \caption{An example subproblem $(c_3, p_1, l_1, p_2, l_2)$.}
    \label{fig:subproblem}
\end{figure}
 The recursion is as follows. For each Hanan grid point $p\in c_j$ and $i = 1, \ldots, k$
\begin{align}
\begin{split}
\label{eqn:recursion_DP}
    &(c_j, p_1, l_1, \ldots, p_i:=p, l_i, \ldots, p_k, l_k)\\ \hspace{-0.5cm}= &\bigvee\limits_{p'}(c_{j-1},p_1, l_1, \ldots, p_i:=p', l_i - |L1_P(p, p')|, \ldots, p_k, l_k)
    \end{split}
\end{align}
where $p'$ is taken from the set of all Hanan grid points on the cuts $c_1, \ldots, c_{j-1}$ such that geodesic L1 shortest paths from $p'$ to $c_j$ make contact with $c_j$ at $p$ (Lemma~\ref{lem:discrete_viewpoints}).
The base case is simply $(s, s,0,\ldots,s,0) = \text{TRUE}$. After tabulating all subproblems, we take the subproblem $(c_m, p_1, l_1, \ldots, p_k, l_k)$ (such that $(c_m, p_1, l_1, \ldots, p_k, l_k) = $ TRUE) with the minimum $\max\limits_{i=1,\ldots, k}\{l_i + |L1_P(p_i, s)|\}$ and return the tours $\{\gamma_i := \Gamma_i \cup L1_P(p_i, s)\}$.
 \paragraph{Proof of correctness} Our proof of correctness relies on two arguments:
 \begin{itemize}
     \item Since the paths associated with subproblem $(c_{j-1},p_1, l_1, \ldots, p_i:=p', l_i - |L1_P(p, p')|, \ldots, p_k, l_k)$ visit all essential cuts up to $c_{j-1}$, the paths associated with subproblem $(c_j, p_1, l_1, \ldots, p_i:=p, l_i, \ldots, p_k, l_k)$ also visit all essential cuts up to $c_j$ since $p\in c_j$. By induction, the tours returned hence visit all essential cuts.
     \item $\gamma_i$ consists of geodesic L1 shortest paths between contact points with essential cuts (proof of Lemma~\ref{lem:discrete_viewpoints}). If we identify two consecutive contact points on $\gamma_i$, say $p'$ and $p$ in that order, then the length of the portion of $\gamma_i$ from $s$ to $p$ is $l_i$ if and only if the length of the portion of $\gamma_i$ from $s$ to $p'$ is $l_i - |L1_P(p, p')|$.
 \end{itemize}
\paragraph{Analysis of running time} There are $O(n)$ essential cuts, $O(n)$ Hanan grid points on each cut. Each tour $\gamma_i$ must be no longer than $nD$, where $D$ is the length of the longest edge of $P$, therefore $l_i$ is bounded by $nD$. In total, there are $O[n\cdot n^{2k}\cdot (nD)^k] = O(n^{3k +1}D^k)$ subproblems. We pre-compute geodesic L1 shortest paths between Hanan grid points, as well as between Hanan grid points and essential cuts, which equates to solving the \textsc{All Pairs Shortest Path} problem in the embedded graph of the Hanan grid. Then, we can solve each subproblem by iterating through at most $O(n^2)$ previously solved subproblems. Thus, the total running time is $O(n^{3k+3}D^k)$, which is pseudopolynomial for fixed $k$. This is in congruence with the weak NP-hardness from \textsc{Partition}, for which there exists a pseudopolynomial (polynomial in the number of input integers and the largest input integer) time algorithm. A tighter time bound is $O(n^{2k+3}L^k)$, where $L$ is the length of a shortest single orthogonal watchman route of $P$, which is computable in $O(n)$ time if $P$ is simple and orthogonal~\cite{CN}. Clearly $L \le |\partial P| \le nD$ and $\max\limits_{i=1,\ldots,k}|\gamma_i| \le L$ (one shortest single watchman route and $k-1$ routes of length 0 is a feasible solution to the $k$-WRP). In addition, we need not consider any $L1_P(p, p')$ whose length is greater than $L$ for recursion \eqref{eqn:recursion_DP} of the dynamic programming.

\paragraph{Fully polynomial-time approximation scheme} To achieve fully polynomial running time for fixed $k$, we bound the number of subproblems by ``bucketing'' the lengths of paths in $P$. Let $\{\gamma_i\}$ be an optimal collection of $k$ routes. Consider that $L \le \sum\limits_{i=1, \ldots, k}|\gamma_i|$ (the concatenation of $\{\gamma_i\}$ can be considered a single watchman route) hence
\begin{align}
\label{inq:lower_bound}
    \frac{L}{k}\le \max\limits_{i=1,\ldots,k}|\gamma_i|\le L.
\end{align}

Given any $\varepsilon > 0$, we divide $L$ into $\lceil\frac{nk}{\varepsilon}\rceil$ uniform intervals, each no longer than $\frac{\varepsilon L}{nk}$. The length of any geodesic L1 shortest path we take into consideration for recursion \eqref{eqn:recursion_DP} must fall into one of the intervals, we round it down to the nearest interval endpoint. Then, apply the dynamic programming algorithm to the new instance with subproblems defined instead by intervals' endpoints. Let the solution returned be $\{\gamma_i'\}$. For clarity, we denote by $d(.)$ distance/length in the ``rounded down'' instance. Then
\begin{align}
    \label{inq:definition}\max\limits_{i=1,\ldots,k}|\gamma_i| \ge \max\limits_{i=1,\ldots,k}d(\gamma_i) \ge \max\limits_{i=1,\ldots,k}d(\gamma_i').\end{align}
The first inequality follows simply from the fact that we round down any distance from the original instance, the second inequality is by definition, since $\{\gamma_i'\}$ is an optimal solution of the new instance. Now, any route in $\{\gamma_i'\}$ must consist of at most $n$ geodesic L1 shortest paths between Hanan grid points on essential cuts, the length of each differs by no more than $\frac{\varepsilon L}{nk}$ between the original instance and the ``rounded down'' instance. Thus, for any $i$
\[|\gamma_i'| - d(\gamma_i') \le n\cdot\frac{\varepsilon L}{nk} \]
therefore
\begin{align}
\label{inq:rounding}\max\limits_{i=1,\ldots,k}d(\gamma_i') + \frac{\varepsilon L}{k} \ge \max\limits_{i=1,\ldots,k}|\gamma_i'|.\end{align}
Combining all three inequalities \eqref{inq:lower_bound}, \eqref{inq:definition}, \eqref{inq:rounding},  we get 
\[(1 + \varepsilon)\max\limits_{i=1,\ldots,k}|\gamma_i| \ge \max\limits_{i=1,\ldots,k}|\gamma_i'|\]
with a running time of $O\left(n^{2k+3}\left(\frac{nk}{\varepsilon}\right)^k\right)$.

\paragraph{Remark} The FPTAS for orthogonal movement (L1 distance) gives a polynomial time $(\sqrt{2} + \varepsilon)$-approximation to unrestricted movement (L2 distance). 

\begin{theorem}
    For any fixed $k$, the anchored $k$-WRP in a simple orthogonal polygon has an FPTAS for the L1 metric and a polynomial-time $(\sqrt{2} + \varepsilon)$-approximation for the L2 metric. 
\end{theorem}

\section{Quota $k$-Watchmen in a Simple Polygon}
\label{sect:kquotawatchman}

In this section, we assume $P$ is a general simple polygon and the watchmen can move in any direction within $P$.
\paragraph{Constant factor polynomial-time approximation} Let $\{\gamma_i\}$ be an optimal collection of quota $k$-watchman routes to achieve the visibility area quota of $A$ and let $OPT = \max\limits_{i=1,\ldots,k}|\gamma_i|$. Denote by $C_g(r)$ the geodesic disk of radius $r$ centered at $s$, i.e. the locus of all points within geodesic distance (length of the geodesic shortest path) of $r$ from $s$. Let $r=r_{\min}$, where $r_{\min}$ is the smallest value of $r$ such that $|V(C_g(r))| \ge A$; $r_{\min}$ can be computed in $O(n^2\log n)$ time using the ``visibility wave'' methods in~\cite{quickest}. Clearly, $r_{\min} \le \frac{OPT}{2}$, since $C_g\left(\frac{OPT}{2}\right)$ encloses $\{\gamma_i\}$ and must see an area no smaller than $A$. If we repeatedly multiply $r$ by 2, at some point we must have $\frac{r}{2} \le \frac{OPT}{2} \le r$, suppose we have reached this point. Then, $C_g(r)$ contains $\{\gamma_i\}$. Let $\gamma$ be a shortest single route contained within $C_g(r)$ such that $|V(\gamma)|\ge A$ (note that $\gamma$ is not necessarily the shortest single quota watchman route overall in $P$). 

\begin{lemma}
   $ \frac{|\gamma|}{k}\le OPT \le |\gamma|.$
\end{lemma}
\begin{proof}
    Recall that in Section~\ref{sect:kwatchman}, we proved two similar inequalities for watchman routes with orthogonal movement seeing the whole polygon. The same holds here since orthogonality and quota did not play a part in the argument.
\end{proof}
We show how to approximate $\gamma$ (it is NP-hard to exactly compute $\gamma$), and that the number of times we multiply $r$ by 2 is polynomial in $n$.

\begin{lemma}
\label{lem:budget_watchman}
    \cite[Section 3]{huynh_et_al:LIPIcs.SWAT.2024.27} Given a budget $B \ge 0$ and any $\varepsilon > 0$, there exists an $O\left(\frac{n^5}{\varepsilon^6}\right)$ algorithm that computes a route of length at most $(1 + \varepsilon)B$ seeing as much area as any route of length $B$ within $C_g(r)$.
\end{lemma}
We briefly describe the algorithm, and refer the readers to~\cite{huynh_et_al:LIPIcs.SWAT.2024.27} for more details. First, triangulate $P$, including $s$ as a vertex of the triangulation. Then, overlay onto the triangulation a regular square grid of side lengths $\delta = O\left(\frac{\varepsilon B}{n}\right)$ within an axis aligned square of size $B$-by-$B$ centered at $s$. We consider the set of (convex) cells that overlap (both fully and partially) with $C_g(r)$ and their vertices, $S_{\delta, r}$. Let $\gamma_B$ be the $B$-length route within $C_g(r)$ that achieves the most area of visibility. There exists a route of length at most $(1 + \varepsilon)B$ with vertices coming from $S_{\delta,r}$ enclosing $\gamma_B$, i.e. the boundary of the relative convex hull (the minimum-perimeter connected superset within $P$, see \cite{huynh_et_al:LIPIcs.SWAT.2024.27, mitchell2000geometric}) of the cells containing vertices of $\gamma_B$, thus seeing at least as much area as $\gamma_B$ (Figure~\ref{fig:approx_gamma_prime}). If $|\gamma| \le B\le \alpha|\gamma|$ for some $\alpha \ge 1$, using dynamic programming, the algorithm in~\cite[Section 3]{huynh_et_al:LIPIcs.SWAT.2024.27} computes a route $\gamma'$ of length no longer than $\alpha(1 + \varepsilon)|\gamma|$ with vertices in $S_{\delta,r}$ that sees the most area, which must be no smaller than $|V(\gamma_B)| \ge |V(\gamma)| \ge A$. Using Lemma~\ref{lem:bounding_gamma_prime}, we acquire a polynomial-sized set of values from which we can search for an appropriate $B$\@.
\begin{figure}[h]
    \centering
    \includegraphics[width=0.5\textwidth]{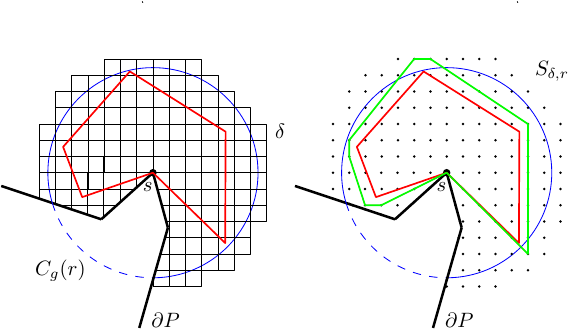}
    \caption{Left: $\gamma_B$ (red) is a tour no longer than $B$ within $C_g(r)$ (blue) that sees the most area. Right: enclosing $\gamma_B$ with a tour whose vertices are in $S_{\delta, r}$ seeing everything $\gamma_B$ sees (green).}
    \label{fig:approx_gamma_prime}
\end{figure}
\begin{lemma}
\label{lem:bounding_gamma_prime}
    $r_{\min}\le |\gamma| \le 9nr_{\min}$.
\end{lemma}
\begin{proof}
     The first inequality is straightforward, $r_{min}\le OPT \le |\gamma|$.
    
    For the second inequality, first note that if a single watchman travels from $s$ to $\partial C_g(r_{min})$, follows along the whole of $\partial C_g(r_{min})$ then goes back to $s$, he sees an area of $A$, thus $|\partial C_g(r_{min})| + 2r_{min}$ is an upper bound on $|\gamma|$. We show that $|\partial C_g(r_{min})| + 2r_{\min} \le 9nr_{min}$. Observe that $\partial C_g(r_{min})$ consists of polygonal chains that are portions of $\partial P$ and at most $n$ circular arcs; the circular arcs have total length no greater than $2n\pi r_{min}$. Each segment in the polygonal part of $\partial C_g(r_{min})$ has length bounded by the sum of geodesic distances from its endpoints to $s$ (triangle inequality), which is no more than $2r_{min}$. There are at most $n$ segments in the polygonal portions of $\partial C_g(r_{min})$, therefore their total length is no greater than $2nr_{\min}$, implying $|\partial C_g(r_{min})| + 2r_{min} = 2nr_{min} + 2n\pi r_{min} + 2r_{min} \le 9nr_{min}$. 
\end{proof}
We divide $9nr_{min}$ into $\lceil\frac{9n}{\varepsilon}\rceil$ uniform intervals so that each is no longer than $\varepsilon r_{min}$: the smallest interval endpoint that is no smaller than $|\gamma|$ must also be no larger than $(1 + \varepsilon)|\gamma|$, and hence is the value of $B$ that we desire. We perform a binary search on the values $\left\{0, \frac{9nr_{min}}{\lceil\frac{9n}{\varepsilon}\rceil}, \ldots, 9nr_{min} \right\}$ as the input budget for the algorithm in Lemma~\ref{lem:budget_watchman}, and seek out the smallest value for which the output route $\gamma'$ sees an area no smaller than $A$\@. Clearly, $|\gamma'|\le (1+\varepsilon)^2|\gamma|$.

Lemma~\ref{lem:bounding_gamma_prime} also implies that the number of times we double $r$ is polynomially bounded, in particular, $O(\log n)$, since $r_{min} \le OPT \le 9nr_{min}$.

We are now ready to describe the approximation algorithm for the Q$k$-WRP as follows:
\begin{itemize}
    \item Step 1: Set $r:= r_{\min}$.
    \item Step 2: Compute $\gamma'$, a $(1 + \varepsilon)^2$-approximation to $\gamma$.
    \item Step 3: Divide $\gamma'$ into $k$ subpaths of equal length, each of which is bounded by $a_i, a_{i+1}\in \gamma'$ ($a_1 \equiv s \equiv a_{k+1}$) and denoted by $\gamma'_{a_ia_{i+1}}$.
    \item Step 4: For each $i$, we obtain $\gamma_i'$ by traversing the geodesic shortest path from $s$ to $a_i$, $\gamma'_{a_ia_{i+1}}$ and the geodesic shortest path from $a_{i+1}$ back to $s$.
    \item Step 5: Set $r:=2r$, then repeat from Step 2, until $r > 9nr_{min}$.
\end{itemize}
Finally, we return the collection of routes $\{\gamma_i'\}$ that minimizes $\max\limits_{i = 1, \ldots, k}|\gamma_i'|$ out of all collections from all values of $r$ in the doubling search. 

\paragraph{Analysis of running time} For each choice of $B$, we execute the $O\left(\frac{n^5}{\varepsilon^6}\right)$ algorithm, thus computing an approximation to $\gamma'$ for each value of $r$ takes $O\left(\frac{n^5}{\varepsilon^6}\log\left(\frac{n}{\varepsilon}\right)\right)$ time. This step dominates both computing $r_{\min}$ and deriving the collection $\{\gamma_i'\}$. Since there are $O(\log n)$ iterations of the doubling search for $r$, the overall running time is $O\left(\frac{n^5}{\varepsilon^6}\log\left(\frac{n}{\varepsilon}\right)\log n\right)$.

\begin{theorem}
    The algorithm described above has an approximation factor of $3 + \varepsilon$.
\end{theorem}
\begin{proof}
    Since all our choices for $B$ are no larger than $9nr_{\min}$, we can choose an appropriate $\delta = O\left(\frac{\varepsilon B}{n}\right)$ so that the geodesic distance from any point on $\gamma'$ to $s$ is no longer than $r + \varepsilon r$. Thus, when $\frac{r}{2} \le \frac{OPT}{2} \le r$, any one of the $k$ routes returned by the algorithm is no longer than $\frac{|\gamma'|}{k} + 2r + 2\varepsilon r \le [(1 + \varepsilon)^2 + 2 + 2\varepsilon]OPT = (3 + \varepsilon')OPT$, where $\varepsilon' = 4\varepsilon + \varepsilon^2$. Note that $\frac{1}{\varepsilon} = \Theta\left(\frac{1}{\varepsilon'}\right)$ as $\varepsilon$ and $\varepsilon'$ approach 0, so the running time is in the same order when written in terms of $\varepsilon'$.
\end{proof}

\paragraph{Improving the approximation factor} In the approximation algorithm earlier, we gradually expand $C_g(r)$ until $C_g(r)$ contains an optimal $\{\gamma_i\}$. If in each iteration, we instead multiply $r$ by a smaller factor, namely $(1 + \varepsilon)$, then at some point $\frac{r}{(1 + \varepsilon)} \le \frac{OPT}{2} \le r$. The distance from each point on $\gamma'$ to $s$ is then no greater than $r + \varepsilon r \le (1 + \varepsilon)^2\frac{OPT}{2}$. Hence, the length of any of the $k$ routes returned by the approximation algorithm is bounded by $\frac{|\gamma'|}{k} + (1 + \varepsilon)^2\frac{OPT}{2} + (1 + \varepsilon)^2\frac{OPT}{2} \le (2 + \varepsilon')OPT$, where $\varepsilon' = 4\varepsilon + 2\varepsilon^2$.

There is however, a trade-off between the approximation factor and the number of iterations of the multiplicative search for $r$. If we multiply $r$ by $(1 + \varepsilon)$ each time, the search requires $O(\log_{1+\varepsilon}n)$ iterations. Note that

\[\log_{1+\varepsilon}n = \log n \frac{\ln2}{\ln(1+\varepsilon)} = \log nO\left(\frac{1}{\varepsilon}\right).\]

In summary, we can achieve an approximation ratio of $(2 + \varepsilon')$ with a running time of $O\left(\frac{n^5}{\varepsilon^7}\log\left(\frac{n}{\varepsilon}\right)\log n\right) = O\left(\frac{n^5}{\varepsilon'^7}\log\left(\frac{n}{\varepsilon'}\right)\log n\right)$ (since $\frac{1}{\varepsilon} = \Theta\left(\frac{1}{\varepsilon'}\right)$). 
\bibliographystyle{plainurl}
\bibliography{refs.bib} 
\end{document}